\documentclass[onecolumn,12pt,draftcls]{IEEEtran}
\IEEEoverridecommandlockouts

\usepackage{amsmath}
\usepackage{amssymb}
\usepackage{amsthm}
\usepackage{dsfont}
\usepackage{enumitem}
\usepackage{bm}
\usepackage{hyperref}

\usepackage{graphicx}
\usepackage{xcolor}

\newcommand{\ind}[1]{\mathds{1}\{#1\}}
\newcommand{\E}[3]{\mathbb{E}^{#2} _{#1}  \left[ #3 \right]}
\newcommand{\Var}[3]{\mathrm{Var}^{#2} _{#1} \left( #3 \right)}
\newcommand{\Prob}[3]{\mathbb{P}^{#2} _{#1} \left\{ #3 \right\}}

\newcommand{\floor}[1]{\left\lfloor #1 \right\rfloor}

\newcommand{\abs}[1]{\left| #1 \right|}
\newcommand{\R}{\mathbb{R}}
\newcommand{\mc}{\mathcal}

\newcommand{\td}{\Tilde}
\newcommand{\KL}[2]{D(#1||#2)}
\newcommand{\eps}{\epsilon}

\newcommand{\FAR}[2]{\mathrm{FAR} ^{#1} \left( #2 \right)}
\newcommand{\WADD}[2]{\mathrm{WADD} ^{#1} \left( #2 \right)}

\DeclareMathOperator*{\esssup}{ess\,sup}

\newtheorem{theorem}{Theorem}[section]

\newtheorem{lemma}[theorem]{Lemma}

\theoremstyle{definition}
\newtheorem{definition}{Definition}[section]

\theoremstyle{remark}
\newtheorem*{remark}{Remark}    

\begin{document}

\title{Non-Parametric Quickest Mean Change Detection
}

\author{Yuchen Liang, ~\IEEEmembership{Student Member,~IEEE}, and  Venugopal V. Veeravalli, ~\IEEEmembership{Fellow, ~IEEE}

\thanks{The material in this paper was presented in part at the Conference on Information Sciences and Systems (CISS) (online) in 2021  \cite{ciss2021}.}

\thanks{This work  was supported in part by the National Science Foundation under grant  ECCS-2033900, and by the Army Research Laboratory under Cooperative Agreement W911NF-17-2-0196, through the University of Illinois at Urbana-Champaign.}

\thanks{The authors are with the ECE Department
of the University of Illinois at Urbana-Champaign. Email: \{yliang35, vvv\}@ILLINOIS.EDU. 
}
}
\maketitle

\vspace*{-0.5in}

\begin{abstract}
The problem of quickest detection of a change in the mean of a sequence of independent observations is studied. The pre-change distribution is assumed to be stationary, while the post-change distributions are allowed to be non-stationary. The case where the pre-change distribution is known is studied first, and then the extension where only the mean and variance of the pre-change distribution are known. No knowledge of the post-change distributions is assumed other than that their means are above some pre-specified threshold larger than the pre-change mean. For the case where the pre-change distribution is known, a test is derived that asymptotically minimizes the worst-case detection delay over all possible post-change distributions, as the false alarm rate goes to zero. Towards deriving this asymptotically optimal test, some new results are provided for the general problem of asymptotic minimax robust quickest change detection in non-stationary settings.  Then, the limiting form of the optimal test is studied as the gap between the pre- and post-change means goes to zero, called the Mean-Change Test (MCT). It is shown that the MCT can be designed with only knowledge of the mean and variance of the pre-change distribution. The  performance of the MCT is also characterized when the mean gap is moderate, under the additional assumption that the distributions of the observations have bounded support. The analysis is validated through numerical results for detecting a change in the mean of a beta distribution. The use of the MCT in monitoring pandemics is also demonstrated.

\end{abstract}


\section{Introduction}
Quickest change detection (QCD) is a fundamental problem in mathematical statistics (see, e.g., \cite{vvv_qcd_overview} for an overview). Given a stochastic sequence whose distribution changes at some unknown change-point, the goal is to detect the change after it occurs as quickly as possible, subject to false alarm constraints. The QCD framework has seen a wide range of applications, including line-outage in power systems \cite{line-outage}, dim-target manoeuvre detection \cite{Molloy2017}, stochastic process control \cite{spc}, structural health monitoring \cite{shm}, and piece-wise stationary multi-armed bandits \cite{kaufmann}. The two main formulations of the classical QCD problem are the Bayesian formulation \cite{Shiryaev:1963,bayesqcd}, where the change-point is assumed to follow a known prior distribution, and the minimax formulation \cite{lorden1971,pollak1987}, where the worst-case detection delay is minimized over all possible change-points, subject to false alarm constraints. In both the Bayesian and minimax settings, if the pre- and post-change distributions are known, low-complexity efficient solutions to the QCD problem can be found \cite{vvv_qcd_overview}.

In many practical situations, we may not know the exact distribution in the pre- or post-change regimes. While it is reasonable to assume that we can obtain a large amount of data in the pre-change regime, this may not be the case for the post-change regime. Also, in applications such  epidemic monitoring and piece-wise stationary multi-armed bandits, a change in a specific statistic (e.g., the mean) of the distribution is of interest. This is different from the original QCD problem where any distributional change needs to be detected. Furthermore, in many applications, the support of the distribution is bounded. For example, the observations representing the fraction of some specific group in the entire population are bounded between 0 and 1. This is the case, for example, in the pandemic monitoring problem that we discuss in detail in Section \ref{sec:num-res}. In many applications, including the pandemic monitoring problem, the system has usually reached some nominal steady-state distribution before the change-point. In these situations, the pre-change distribution can be assumed to be stationary. 

In this paper, we study the problem of quickest detection of a change in the mean of a sequence of independent observations.
The pre-change distribution is assumed to be stationary, while the post-change distributions are allowed to be non-stationary. We first study the case where the pre-change distribution is known, and then study the extension where only the mean and variance of the pre-change distribution are known. No knowledge of the post-change distributions is assumed other than that their means are above some threshold larger than the pre-change mean.

There have been a number of lines of work on the QCD problem when the pre- and/or post-change distributions are not completely known. The most prevalent is the generalized likelihood ratio (GLR) approach,  introduced in \cite{lorden1971} for the parametric case where the post-change distribution has an unknown parameter. This GLR approach is studied in detail for the problem of detecting the change in the mean of a Gaussian distribution with unknown post-change mean in \cite{siegmund1995}. A GLR test for the case where the pre- and post-change distributions come from an one-parameter exponential family, and both the pre- and post-change parameters are unknown, is analyzed in \cite{LaiXing2010}.

The QCD problem has also been studied in a non-parametric setting. In particular, for detecting a change in the mean of an observation sequence, one approach has been to use maximum \emph{scan statistics}. The scan statistic of an observation sequence is defined as the absolute difference of the averages before and after a potential change-point. In \cite{dark-mean-change}, the case where the pre- and post-change distributions have finite moment generating functions in some neighborhood around zero is considered. 
At each time greater than a window size $N$, the scan statistic at each potential change-point is calculated using the last $N$ observations. The maximum scan statistic is then calculated over the set of potential change-points, and an alarm is raised if this maximum exceeds some threshold. 
In \cite{maillard19a}, the case of sub-Gaussian pre- and post-change distributions is studied. The scan statistic is calculated over the entire observation sequence, and the maximum is compared to a threshold determined by the current time and the desired false alarm rate. This approach is further applied to the piece-wise stationary multi-armed bandit problem in \cite{kaufmann}. We compare our approach to mean-change detection with a test using scan statistics in Section \ref{sec:num-res}.

We note that for both the GLR  the scan statistics approaches, the complexity of computing the test statistic at each time-step grows at least linearly with the number of samples. In practice, a windowed version of the test statistic is often used to reduce computational complexity, while suffering some loss in performance. 

Still another line of work is the one based on a \textit{minimax} robust approach \cite{huber1965}, in which it is assumed that the distributions come from mutually exclusive uncertainty classes. Under certain conditions on the uncertainty classes, e.g., joint stochastic boundedness \cite{moulin-veeravalli-2018},  low-complexity solutions to the minimax robust QCD problem can be found \cite{Unnikrishnan2011}. Under more general conditions, e.g., weak stochastic boundedness, a solution that is asymptotically close to the minimax solution can be found \cite{Molloy2017}.

In this paper, we use an asymptotic version of the minimax robust QCD problem formulation \cite{Molloy2017} to develop algorithms for the non-parametric detection of a change in mean of an observation sequence. Our contributions are as follows:
\begin{enumerate}
    \item We extend the asymptotic minimax robust QCD problem introduced in \cite{Molloy2017} to the more general non-stationary setting.
    
    \item We study the problem of quickest detection of a change in the mean of an observation sequence under the assumption that no knowledge of the post-change distribution is available other than that 
    its mean is above some threshold larger than the pre-change mean.

    \item For the case where the pre-change distribution is known, we derive a test that asymptotically minimizes the worst-case detection delay over all possible post-change distributions, as the false alarm rate goes to zero.
    
    \item We study the limiting form of the optimal test as the gap between the pre- and post-change means goes to zero, which we call the Mean-Change Test (MCT). We show that the MCT can be designed with only knowledge of the mean and variance of the pre-change distribution.
    
    \item We also characterize the performance of the MCT when the mean gap is moderate, under the assumption that the distributions of the observations have bounded support. 

    \item We validate our analysis through numerical results for detecting a change in the mean of a beta distribution. We also demonstrate the use of the MCT for pandemic monitoring.
\end{enumerate}

The rest of the paper is structured as follows. In Section~\ref{sec:qcd-distuc}, we describe the quickest change detection problem under distributional uncertainty and provide some new results regarding asymptotically robust tests in the non-stationary setting. In Section~\ref{sec:mc-prob}, we formulate the mean change detection problem, and propose and analyze the mean-change test (MCT), which solves the problem asymptotically.
In Section \ref{sec:num-res}, we validate our analysis through numerical results for detecting a change in the mean of a beta distribution, and also demonstrate the use of the MCT in monitoring pandemics. Finally, in Section~\ref{sec:concl}, we provide some concluding remarks.

\section{Quickest Change Detection Under Distributional Uncertainty}
\label{sec:qcd-distuc}

Let $X_1,\dots,X_t,\dots \in \R$ be a sequence of independent random variables, and let $\nu$ be a change-point. Let $\bm{P_0}=\{P_{0,t}\}_{t \geq 1}$ and $\bm{P_1}=\{P_{1,t}\}_{t \geq 1}$ be two sequences of probability measures, where $P_{0,t} \in {\cal P}_0$ and $P_{1,t} \in {\cal P}_1$ for all $t \geq 1$. Further, assume that $P_{j,t}$ has  probability density $p_{j,t}$ with respect to the Lebesgue measure on $\R$, for $j=0,1$ and $t \geq 1$. Let $\Prob{\nu}{\bm{P_0},\bm{P_1}}{\cdot}$ denote the probability measure on the entire sequence of observations when the pre-change distributions are $\{P_{0,t}\}_{t < \nu}$ and the post-change distributions are $\{P_{1,t}\}_{t \geq \nu}$,  with $X_t \sim P_{0,t},\ \forall 1 \leq t < \nu$ and $X_t \sim P_{1,t},\ \forall t \geq \nu$, and let $\E{\nu}{\bm{P_0},\bm{P_1}}{\cdot}$ denote the corresponding expectation.
When $\bm{P_0}$ and $\bm{P_1}$ are stationary, i.e., $P_{0,t}= P_0$, $\forall t \geq 1$ and $P_{1,t}= P_1$, $\forall t \geq 1$, we use the notations $\Prob{\nu}{P_0,P_1}{\cdot}$ and $\E{\nu}{P_0,P_1}{\cdot}$ in place of $\Prob{\nu}{\bm{P_0},\bm{P_1}}{\cdot}$ and $\E{\nu}{\bm{P_0},\bm{P_1}}{\cdot}$, respectively.


The change-time $\nu$ is assumed to be unknown but deterministic. The problem is to detect the change quickly while not causing too many false alarms. Let $\tau$ be a stopping time \cite{moulin-veeravalli-2018} defined on the observation sequence associated with the detection rule, i.e. $\tau$ is the time at which we stop taking observations and declare that the change has occurred.

For the case where both the pre- and post-change distributions are stationary and known, Lorden \cite{lorden1971} proposed solving the following optimization problem to find the best stopping time $\tau$:
\begin{equation}
\label{orig_qcd}
    \inf_{\tau \in \mathcal{C}^{P_0}_\alpha} \WADD{P_0,P_1}{\tau}
\end{equation}
where
\begin{equation}
\label{wadd1}
    \WADD{P_0,P_1}{\tau} :=\\
    \sup_{\nu \geq 1} \esssup \E{\nu}{P_0,P_1}{\left(\tau-\nu+1\right)^+|X_1,\ldots, X_{\nu-1}}
\end{equation}
is a worst-case delay metric, and
\begin{equation}
    \mathcal{C}^{P_0}_\alpha := \left\{ \tau: \FAR{P_0}{\tau} \leq \alpha \right\}
\end{equation}
with  
\begin{equation}
    \FAR{P_0}{\tau} := \frac{1}{ \E{\infty}{P_0,P_1}{\tau}}.
 \end{equation}   
Here $\E{\infty}{P_0,P_1}{\cdot}$ is the expectation operator when the change never happens, and $(\cdot)^+:=\max\{0,\cdot\}$.

Lorden also showed that Page's Cumulative Sum (CuSum) algorithm \cite{page1954} whose test statistic is given by:
\begin{align}
\label{cusum_stat}
    \Lambda^{P_0,P_1}(t) &= \max_{1\leq k \leq t+1} \sum_{i=k}^t \ln L^{P_0,P_1}(X_i) \nonumber\\
    &= \left(\Lambda^{P_0,P_1}(t-1) + \ln L^{P_0,P_1}(X_t) \right)^+
\end{align}
solves the problem in (\ref{orig_qcd}) asymptotically.
Here $L^{P_0,P_1}$ is the likelihood ratio:
\begin{equation}\label{eq:LRdef}
 L^{P_0,P_1}(x) = \frac{p_1(x)}{p_0(x)}.
\end{equation}
The CuSum stopping rule is given by:
\begin{equation}
\label{defstoppingrule}
    \tau\left(\Lambda^{P_0,P_1}, b_\alpha\right) := \inf \{t:\Lambda^{P_0,P_1}(t)\geq b_\alpha \}
\end{equation}
where $b_\alpha := |\ln \alpha|$. It was shown by Moustakides \cite{moustakides1986} that the CuSum algorithm is exactly optimal for the problem in (\ref{orig_qcd}).

When the pre-change and post-change distributions are unknown but belong to known uncertainty sets and are possibly non-stationary, a minimax robust formulation can be used in place of \eqref{orig_qcd}:
\begin{equation}
\label{newmainprob}
    \inf_{\tau \in {\cal C}_\alpha^{{\cal P}_0}}\quad  \sup_{(\bm{P_0},\bm{P_1}): (P_{0,t},P_{1,t}) \in {\cal P}_0 \times {\cal P}_1 ,\forall t} \WADD{\bm{P_0},\bm{P_1}}{\tau}
\end{equation}
where
\begin{equation}
    \WADD{\bm{P_0},\bm{P_1}}{\tau} :=\\
    \sup_{\nu \geq 1} \esssup \E{\nu}{\bm{P_0},\bm{P_1}}{\left(\tau-\nu+1\right)^+|X_1,\ldots, X_{\nu-1}}
\end{equation}
and the feasible set is defined as
\begin{equation}
\label{main_fa_constraint}
    \mathcal{C}_{\alpha}^{\mc{P}_0} = \left\{ \tau: \sup_{\bm{P_0}: P_{0,t} \in \mathcal{P}_0} \FAR{\bm{P_0}}{\tau} \leq \alpha \right\}
\end{equation}
with  
\begin{equation}
   \FAR{\bm{P_0}}{\tau} := \frac{1}{ \E{\infty}{\bm{P_0},\bm{P_1}}{\tau}}.
 \end{equation}

We now address the solution to the problem in \eqref{newmainprob}. To this end, we give the following using definitions.

\begin{definition} (see, e.g., \cite{moulin-veeravalli-2018}) A pair of uncertainty sets $(\mc{P}_0,\mc{P}_1)$ is said to be \textit{jointly stochastically (JS) bounded}  by $(\Bar{P}_0,\Bar{P}_1) \in \mc{P}_0 \times \mc{P}_1$ if, for any $(P_0,P_1) \in \mc{P}_0 \times \mc{P}_1$ and any $h > 0$,
\begin{align}
    P_0 \{L^{\Bar{P}_0,\Bar{P}_1}(X) > h\} &\leq \Bar{P}_0 \{L^{\Bar{P}_0,\Bar{P}_1}(X) > h\}\\
    P_1 \{L^{\Bar{P}_0,\Bar{P}_1}(X) > h\} &\geq \Bar{P}_1 \{L^{\Bar{P}_0,\Bar{P}_1}(X) > h\}
\end{align}
where $L^{\Bar{P}_0,\Bar{P}_1}$ is the likelihood ratio between $\Bar{P}_1$ and $\Bar{P}_0$ (see \eqref{eq:LRdef}). The distributions $\Bar{P}_0$ and  $\Bar{P}_1$ are called least favorable distributions (LFDs) within the classes ${\cal P}_0$ and ${\cal P}_1$, respectively. 
\end{definition}

If the pair of pre- and post-change uncertainty sets is JS bounded, the CuSum test statistic $\Lambda^{\Bar{P}_0,\Bar{P}_1}(t)$ (see \eqref{cusum_stat}), with stopping rule $\tau(\Lambda^{\Bar{P}_0,\Bar{P}_1},b_\alpha)$ (see \eqref{defstoppingrule}), solves (\ref{newmainprob}) exactly both when $\bm{P_0}$ and $\bm{P_1}$ are stationary \cite{Unnikrishnan2011} and when they are potentially non-stationary \cite{molloy-non-stat}.

\begin{definition} (see \cite{Molloy2017})
A pair of uncertainty sets $(\mc{P}_0,\mc{P}_1)$ is said to be \textit{weakly stochastically (WS) bounded} by $(\td{P}_0,\td{P}_1) \in \mc{P}_0 \times \mc{P}_1$ if
\begin{equation}
\label{wsb:kl}
    \KL{\td{P}_1}{\td{P}_0} \leq \KL{P_1}{\td{P}_0} - \KL{P_1}{\td{P}_1}
\end{equation}
for all $P_1 \in \mc{P}_1$, and
\begin{equation}
\label{wsb:exp}
    \E{}{P_0}{L^{\td{P}_0,\td{P}_1}(X)} \leq \E{}{\td{P}_0}{L^{\td{P}_0,\td{P}_1}(X)} = 1
\end{equation}
for all $P_0 \in \mc{P}_0$. Here, $\E{}{P}{\cdot}$ denotes the expectation operator with respect to distribution $P$, and $\KL{P}{Q}$ denotes KL-divergence:
\begin{equation}\label{eq:KLdef}
\KL{P}{Q} = \E{}{P}{\ln L^{P,Q} (X)}.
\end{equation}
\end{definition}

It is shown in \cite{Molloy2017} that if the pair of uncertainty sets is JS bounded by $(\Bar{P}_0,\Bar{P}_1)$, it is also WS bounded by $(\Bar{P}_0,\Bar{P}_1)$. It is also shown in \cite{Molloy2017} that if the pair of pre- and post-change uncertainty sets is WS bounded, the CuSum test statistic $\Lambda^{\td{P}_0,\td{P}_1}(t)$ with stopping rule $\tau(\Lambda^{\td{P}_0,\td{P}_1},b_\alpha)$ solves (\ref{newmainprob}) asymptotically as $\alpha \to 0$ when $\bm{P_0}$ and $\bm{P_1}$ are both stationary. 

\subsection{Asymptotically Optimal Solution in the Non-stationary Setting}

Let $\td{P}_0,\td{P}_1$ be such that ${\cal P}_0 \times {\cal P}_1$ is WS bounded by $(\td{P}_0,\td{P}_1)$. In the following, we extend the result in \cite{Molloy2017} to the case where $\bm{P_0}$ and $\bm{P_1}$ are potentially non-stationary and derive an asymptotically optimal solution as $\alpha \to 0$. Specifically, through Lemma \ref{molloy_delay_ext} we upper bound the asymptotic delay, through Lemma \ref{molloy_fa_ext} we control the false alarm rate, and in Theorem \ref{asymp_opt_non_stat} we combine the lemmas to provide an asymptotically optimal solution to the problem in \eqref{newmainprob} when $\bm{P_0}$ and $\bm{P_1}$ are potentially non-stationary.

\begin{lemma}
\label{molloy_delay_ext}
Consider ${\cal P}_0 \times {\cal P}_1$ WS bounded by $(\td{P}_0,\td{P}_1)$. Let $\bm{P_0}$ and $\bm{P_1}$ be such that $P_{0,t} \in {\cal P}_0$ and $P_{1,t} \in {\cal P}_1$ for all $t \geq 1$. Suppose that for all $P_{1,t} \in {\cal P}_1$, 
\begin{equation*}
\sup_{1 \leq t \leq n} \Var{}{P_{1,t}}{\ln L^{\td{P}_0,\td{P}_1} (X_t)} = o(n) \text{ as $n \to \infty$}
\end{equation*}
where $\Var{}{P}{X}$ denotes the variance of $X$ when $X \sim P$. Then, $\tau(\Lambda^{\td{P}_0,\td{P}_1},b)$ satisfies
\begin{equation}
    \WADD{\bm{P_0},\bm{P_1}}{\tau(\Lambda^{\td{P}_0,\td{P}_1},b)} \leq (1+o(1))\left(\frac{b}{\KL{\td{P}_1}{\td{P}_0}}\right)
\end{equation}
as $b \to \infty$, where $o(1) \to 0$ as $b \to \infty$.
\end{lemma}

\begin{lemma}
\label{molloy_fa_ext}
Consider the same assumptions as in Lemma~\ref{molloy_delay_ext}. Then, for any $P_{0,t} \in {\cal P}_0$,
\begin{equation}
    \E{}{\bm{P_0}}{\tau(\Lambda^{\td{P}_0,\td{P}_1},b)} \geq e^b
\end{equation}
for any threshold $b > 0$.
\end{lemma}

\begin{theorem}
\label{asymp_opt_non_stat}
Consider the same assumptions as in Lemma~\ref{molloy_delay_ext}. Then, the CuSum test $\tau(\Lambda^{\td{P}_0,\td{P}_1},b_\alpha)$ solves the problem in (\ref{newmainprob}) asymptotically as $\alpha \to 0$, and
\begin{equation}
    \sup_{(\bm{P_0},\bm{P_1}): (P_{0,t},P_{1,t}) \in {\cal P}_0 \times {\cal P}_1 ,\forall t} \WADD{\bm{P_0},\bm{P_1}}{\tau(\Lambda^{\td{P}_0,\td{P}_1},b_\alpha)} = (1+o(1)) \left(\frac{|\ln \alpha|}{\KL{\td{P}_1}{\td{P}_0}}\right)
\end{equation}
where $o(1) \to 0$ as $\alpha \to 0$.
\end{theorem}

The proofs of Lemma~\ref{molloy_delay_ext}, Lemma~\ref{molloy_fa_ext} and Theorem~\ref{asymp_opt_non_stat} are given in the appendix.

\section{Mean-Change Detection Problem}
\label{sec:mc-prob}

Until now, we have considered the general QCD problem formulated in \eqref{newmainprob}. In this paper, we are mainly interested in a special case of the problem, described as follows.
The pre-change distribution is stationary, i.e., $P_{0,t}=P_0, \forall t \geq 1$, with pre-change mean $\mu_0 = \E{}{P_0}{X}$ and variance $\sigma_0^2 = \Var{}{P_0}{X}$. Thus, $\mc{P}_0 = \{ P_0 \}$ is a singleton.
The post-change distribution could be non-stationary, and at each time it belongs to the following uncertainty set:
\begin{equation}
\label{f1constraint}
    \mathcal{P}_1 = \mathcal{M}_1 := \{ P: \E{}{P}{X} \geq \eta > \mu_0\}.
\end{equation}
In this expression, $X$ denotes a generic observation in the sequence, and $\eta$ is a pre-designed threshold. Define 
\begin{equation}
    \Delta := \frac{\eta-\mu_0}{2}
\end{equation}
which is half of the worst-case mean-change gap.

The minimax robust mean-change problem, which is a reformulation of \eqref{newmainprob} is given by:
\begin{equation}
\label{mean_change_prob}
    \inf_{\tau \in {\cal C}_\alpha^{P_0}}\quad  \sup_{\bm{P_1}: P_{1,t} \in {\cal M}_1 ,\forall t} \WADD{P_0,\bm{P_1}}{\tau}.
\end{equation}
Our goal is to find a stopping time that solves \eqref{mean_change_prob} asymptotically as the false alarm rate $\alpha \to 0$.

\subsection{Known Pre-change Distribution}
\label{subsec:known_pre}

Define
\begin{equation}
\label{eq:kappa0_def}
    \kappa_0 (\lambda) = \ln \E{}{P_0}{e^{\lambda X}}   
\end{equation} 
to be the cumulant-generating function (cgf) of the observations under $P_0$. In the following theorem, we provide a solution to the problem stated in \eqref{mean_change_prob}.

\begin{theorem}
Consider ${\cal P}_0 = \{P_0\}$, and ${\cal M}_1$ as given in (\ref{f1constraint}). Define
\begin{equation}
\label{tilteddistr}
    \td{p}_1(x) = p_0(x) e^{\lambda^* x - \kappa_0 (\lambda^*)}
\end{equation}
where $\kappa_0 (\lambda)$ is the cgf under $P_0$ and $\lambda^*$ satisfies
\begin{equation}
\label{lambdastar}
    \kappa_0'(\lambda^*) = \frac{\E{}{P_0}{X e^{\lambda^* X}}}{\E{}{P_0}{e^{\lambda^* X}}} = \eta
\end{equation}
Then, the CuSum statistic
\begin{equation}
\label{opt_stat}
    \Lambda^{P_0,\td{P}_1}(t) = \max_{1\leq k \leq t+1} \sum_{i=k}^t \left(\lambda^* X_i - \kappa_0(\lambda^*) \right)
\end{equation}
and the stopping rule $\tau(\Lambda^{P_0,\td{P}_1},b_\alpha)$ (see \eqref{defstoppingrule}) with threshold $b_\alpha=|\ln{\alpha}|$ solves the minimax robust problem in \eqref{mean_change_prob} asymptotically as $\alpha \to 0$, and
\begin{equation}
    \inf_{\tau \in {\cal C}_\alpha^{P_0}}\quad  \sup_{\bm{P_1}: P_{1,t} \in {\cal M}_1 ,\forall t} \WADD{P_0,\bm{P_1}}{\tau} = \frac{|\ln{\alpha}|}{\lambda^* \eta - \kappa_0(\lambda^*)} (1+o(1))
\end{equation}
\end{theorem}


\begin{proof}
The proof follows from an application of Theorem~\ref{asymp_opt_non_stat} if we can establish that $\mc{P}_0 \times {\cal M}_1$ is WS bounded by $(P_0, \td{P}_1)$. By \cite[Prop.~1 (iii)]{Molloy2017}, since ${\cal M}_1$ is convex and $\mc{P}_0$ is a singleton, if $\td{P}_1$   minimizes the KL-divergence $\KL{P_1}{P_0}$ over $P_1 \in {\cal M}_1$, then $\mc{P}_0 \times {\cal M}_1$ is WS bounded by $(P_0, \td{P}_1)$. Therefore, it remains to show that $\td{P}_1$ specified in \eqref{tilteddistr} minimizes $\KL{P_1}{P_0}$, subject to $\E{}{P_1}{X} \geq \eta$. To this end, we follow the procedure outlined in \cite[Sec.~6.4.1]{levy-detection}. Consider the Lagrangian
\begin{align}
    J(p_1,\lambda,\mu) &= \E{}{P_1}{\ln L^{P_0,P_1}(X)} + \lambda (\eta - \E{}{P_1}{X}) + \mu \left(1 - \int p_1(x) d d x \right) \nonumber\\
    &= \int \left(\ln \frac{p_1(x)}{p_0(x)} - \lambda x - \mu \right) p_1(x) d x + \lambda \eta + \mu 
\end{align}
where the Lagrange multiplier $\lambda \geq 0$ corresponds to the constraint that the post-change mean is greater than $\eta$, and $\mu$ corresponds to the constraint that $p_1(x)$ is a probability measure. 
For an arbitrary direction $z$, we take the Gateaux derivative with respect to $p_1$:
\begin{align}
    \nabla_{p_1,z} J(p_1,\lambda,\mu) &:= \lim_{h \to 0} \frac{J(p_1+h z,\lambda,\mu)-J(p_1,\lambda,\mu)}{h} \nonumber\\
    &= \int \left(\ln \frac{p_1(x)}{p_0(x)} - \lambda x - \mu'\right) z d x
\end{align}
where $\mu' = \mu - 1$, and since $z$ is arbitrary, we arrive at
\begin{equation}
    \ln \frac{p_1(x)}{p_0(x)} - \lambda x - \mu' = 0
\end{equation}
By the Generalized Kuhn–Tucker Theorem \cite{luenburger-opt-by-vec-space}, since $p_0 (x)$ is bounded, $p_1 (x) = p_0(x) e^{\lambda x+\mu'}$ is a necessary condition for optimality. Furthermore, since $J(p_1,\lambda,\mu)$ is convex in $p_1$, this is also a global optimum. To satisfy the constraints, we have
\begin{equation}
    \mu' = -\ln \int p_0 (x) e^{\lambda x} d x = - \kappa_0 (\lambda)
\end{equation}
and that $\lambda^*$ satisfies
\begin{equation}
    \eta = \E{}{P_1}{X} = \frac{\E{}{P_0}{X e^{\lambda^* X}}}{\E{}{P_0}{e^{\lambda^* X}}} = \kappa_0'(\lambda^*)
\end{equation}
Thus, $\td{P}_1$ in (\ref{tilteddistr}) minimizes $\KL{P_1}{P_0}$, subject to $\E{}{P_1}{X} \geq \eta$.

Furthermore, the minimum KL-divergence is
\begin{align}
    \KL{\td{P}_1}{P_0} &= \int (\lambda^* x - \kappa_0(\lambda^*)) \td{p}_1 (x) d x \nonumber\\
    &= \lambda^* \eta - \kappa_0(\lambda^*)
\end{align}
Hence, the worst-case delay satisfies
\begin{align}
    \inf_{\tau \in {\cal C}_\alpha^{P_0}}\quad  \sup_{\bm{P_1}: P_{1,t} \in {\cal M}_1 ,\forall t} \WADD{P_0,\bm{P_1}}{\tau} &= \frac{|\ln{\alpha}|}{\KL{\td{P}_1}{P_0}}(1+o(1)) \nonumber\\
    &= \frac{|\ln{\alpha}|}{\lambda^* \eta - \kappa_0(\lambda^*)}(1+o(1))
\end{align}
as $\alpha \to 0$. \qedhere
\end{proof}

Note that $\td{p}_1$ is an exponentially-tilted version (or the Esscher transform) of $p_0$.

\subsection{Approximation for Small $\Delta$}
\label{subsec:approx_small_delta}

Even though we have an expression for the test statistic when $P_0$ is known, as given in (\ref{opt_stat}), the exact solution of $\lambda^*$ is not available in closed-form. Fortunately, if the mean-change gap $\Delta$ is small, we obtain a low-complexity test in terms of only the pre-change mean and variance that closely approximates the performance of the asymptotically minimax optimal test in the previous section.

As $\Delta \to 0$, $\eta \to \mu_0$, and hence $\lambda^* \to 0$. From a second-order Taylor expansion on $\kappa_0$ around 0, we obtain
\begin{align}
\label{taylor}
    \kappa_0(\lambda^*) &= \kappa_0(0) + \kappa_0'(0) \lambda^* + \frac{\kappa_0''(0)}{2} (\lambda^*)^2 + o((\lambda^*)^2) \nonumber\\
    &= \mu_0 \lambda^* + \frac{\sigma_0^2}{2} (\lambda^*)^2 + o((\lambda^*)^2)
\end{align}
In this same regime, by continuity of $\kappa_0'(\cdot)$,
\begin{align}
    \lambda^* &= \frac{\kappa_0'(\lambda^*) - \kappa_0'(0)}{\kappa_0''(0)} + o(\Delta) \nonumber\\
    &= \frac{\eta - \mu_0}{\sigma_0^2} + o(\Delta) \nonumber\\
    &= \frac{2 \Delta}{\sigma_0^2} + o(\Delta)
\end{align}
where we have used $\kappa_0'(\lambda^*) = \eta$. Hence, the approximate log-likelihood ratio at time $t$ is
\begin{align}
    \lambda^* X_t - \kappa_0(\lambda^*) &= \lambda^* X_t - (\mu_0 \lambda^* + \frac{\sigma_0^2}{2} (\lambda^*)^2) + o((\lambda^*)^2) \nonumber\\
    &= \frac{2 \Delta}{\sigma_0^2} (X_t - \mu_0) - \frac{\sigma_0^2}{2} \left(\frac{2 \Delta}{\sigma_0^2}\right)^2 + o(\Delta^2) \nonumber\\
    &= \frac{2 \Delta}{\sigma_0^2} \left(X_t - \frac{\mu_0 + \eta}{2}\right) + o(\Delta^2)
\end{align}
and the corresponding minimum KL-divergence is approximated as:
\begin{equation}
\label{kldivapprox}
    \KL{\td{P}_1}{P_0} = \frac{2\Delta^2}{\sigma_0^2} + o(\Delta^2).
\end{equation}
Now
\begin{equation}
    \frac{2 \Delta}{\sigma_0^2} \left(X_t - \frac{\mu_0 + \eta}{2}\right) > b_\alpha \iff X_t - \frac{\mu_0 + \eta}{2} > \td{b}_\alpha
\end{equation}
where
\begin{equation}
\label{balpha}
    \td{b}_\alpha := \frac{|\ln{\alpha}| \sigma_0^2}{2\Delta} = \frac{|\ln{\alpha}| \sigma_0^2}{\eta - \mu_0}.
\end{equation}
Therefore, the stopping rule $\tau(\Lambda^{P_0,\td{P}_1}, b_\alpha)$ can be approximated by the stopping rule $\tau(\td{\Lambda}^{\mu_0,\eta}, \td{b}_\alpha)$, where 
\begin{align}
\label{stat}
    \td{\Lambda}^{\mu_0,\eta} (t) &= \max_{1\leq k \leq t+1} \sum_{i=k}^t \left(X_i - \frac{\mu_0+\eta}{2}\right) \nonumber\\
    &= \left(\td{\Lambda}^{\mu_0,\eta} (t-1)+ \left(X_t - \frac{\mu_0+\eta}{2}\right)\right)^+
\end{align}
with $\td{\Lambda}^{\mu_0,\eta}(0)=0$. We call $\tau(\td{\Lambda}^{\mu_0,\eta}, \td{b}_\alpha)$ the Mean-Change Test (MCT), and  $\td{\Lambda}^{\mu_0,\eta}$ the MCT statistic.

From \eqref{kldivapprox}, it follows that as $\alpha \to 0$ and $\Delta \to 0$, the worst-case delay satisfies
\begin{equation}
\label{est_opt_delay}
    \inf_{\tau \in {\cal C}_\alpha^{P_0}}\quad  \sup_{\bm{P_1}: P_{1,t} \in {\cal M}_1 ,\forall t} \WADD{P_0,\bm{P_1}}{\tau} = \frac{|\ln{\alpha}| \sigma_0^2}{2\Delta^2}(1+o(1)).
\end{equation}

Therefore, if $\Delta$ is small, it is sufficient to know only the mean and variance to construct a good approximation to the asymptotically minimax robust test. Furthermore, only the mean of the pre-change distribution is needed to construct the MCT statistic. From the simulation results in Section IV, we see that the performance of the MCT can be very close to that of the asymptotically minimax robust test even for moderate values of $\Delta$. Since the mean and variance of a distribution are much easier and more accurate to estimate than the entire density, this test can be useful and accurate when only a moderate number of observations in the pre-change regime is available. 


\begin{remark}
It is interesting that the form of MCT statistic in \eqref{stat} coincides with that of the CuSum statistic (see \eqref{cusum_stat}) with known stationary pre- and post-change distributions, $P_0 \sim \mc{N} (\mu_0, \sigma^2)$ and $P_1 \sim \mc{N} (\eta, \sigma^2)$, respectively. Here $\mc{N} (\mu, \sigma^2)$ denotes a Gaussian distribution with mean $\mu$ and variance $\sigma^2$.
\end{remark}

\subsection{Performance Analysis of MCT for moderate $\Delta$}
\label{subsec:mct_analysis}

We now study the asymptotic performance of the MCT for fixed $\Delta$, as $\alpha \to 0$. For this part of the analysis, we assume that the pre- and post-change distributions have supports that are uniformly bounded, and without loss of generality, we assume that the bounding interval is $[0,1]$. This assumption holds in many practical applications, including the pandemic monitoring problem discussed in Section~\ref{sec:num-res}.

Define 
\begin{equation} \label{eq:Zidef}
Z_t := X_t - \frac{\mu_0+\eta}{2},\ \forall t \geq 1.
\end{equation}
Then the MCT statistic of \eqref{stat} can be written as:
\begin{equation}
\label{stat_Z}
    \td{\Lambda}^{\mu_0,\eta} (t) = \left(\td{\Lambda}^{\mu_0,\eta} (t-1)+ Z_t\right)^+
\end{equation}
with $\td{\Lambda}^{\mu_0,\eta}(0)=0$. The MCT stopping time is given by:
\begin{equation}\label{MCT:stopping}
\tau(\td{\Lambda}^{\mu_0,\eta}, b) = \inf \{t:\td{\Lambda}^{\mu_0,\eta} \geq b \}
\end{equation}
where $b$ has to be chosen to meet the FAR constraint:
\begin{equation} \label{eq:FAR_MCT}
\FAR{P_0}{\tau(\td{\Lambda}^{\mu_0,\eta}, b)} = \frac{1}{ \E{\infty}{P_0,\bm{P_1}}{\tau(\td{\Lambda}^{\mu_0,\eta}, b)}} \leq \alpha 
\end{equation}
In what follows, we write $\tau(\td{\Lambda}^{\mu_0,\eta}, b)$ as $\tau(b)$, with the understanding that the test statistic being used throughout is the MCT statistic $\td{\Lambda}^{\mu_0,\eta}$.


\subsubsection{False Alarm Analysis}
In Lemma~\ref{mct:cross_fa} below, we first control the boundary crossing probability of $S_t$ in the pre-change regime. Then, in Theorem~\ref{mct:fa}, we use Lemma~\ref{mct:cross_fa} to bound the false alarm rate of the MCT asymptotically using the procedure outlined in \cite{siegmund_1985}.


\begin{lemma}
\label{mct:cross_fa}
Assume that the pre-change distribution $P_0$ has known pre-change mean $\mu_0$ and variance $\sigma_0^2$, and that the post-change distribution is non-stationary with $P_{1,t} \in {\cal M}_1$, for all $t\geq 1$. For $b > 0$, define the supplementary stopping time
\begin{equation}
\label{tau_prime}
    \tau' (b) := \inf\{t: S_t \notin (0,b) \}
\end{equation}
where $S_t := \sum_{i=1}^t Z_i$, with $Z_i$ defined in \eqref{eq:Zidef}. Then,
\begin{align}
\label{nng:fa}
    \Prob{\infty}{P_0,\bm{P_1}}{S_{\tau'(b)} \geq b} &\leq 2 R_0 \sqrt{\frac{b^2}{\Delta^2}} K_1 \left(\frac{R_0^2 b \Delta }{\sigma_0^2} \right) \exp{\left(- \frac{R_0^2 \Delta}{\sigma_0^2} b \right)} \nonumber\\
    &= \sqrt{\frac{2 \pi \sigma_0^2 b}{\Delta^3}} \exp{\left(- \frac{2 R_0^2 \Delta}{\sigma_0^2} b \right)} (1+o(1)), \text{ as } b \to \infty,
\end{align}
where 
\begin{equation}
\label{rdef}
    R_0=\sigma_0^2 / \left(\sigma_0^2+\Delta \cdot \max\{\mu_0, 1-\mu_0\} / 3\right)
\end{equation}
and $K_\beta (z)$ is the modified Bessel function of the second kind of order $\beta$.
\end{lemma}

\begin{proof}
Note that $\E{}{P_0}{Z_i} = (\mu_0-\eta)/2 = -\Delta$. Since $X_i \in [0,1]$, we have $Z_i + \Delta \in [-\mu_0,1-\mu_0]$. Let $M = \max\{\mu_0/3, (1-\mu_0)/3\}$; then $|Z_i+\Delta| \leq 3 M$. Thus, we have
\begin{align*}
    \Prob{\infty}{P_0,\bm{P_1}}{S_{\tau'(b)} \geq b} &= P_0\left\{S_{\tau'(b)} \geq b\right\} \\
    & = P_0\left\{\sum_{i=1}^{\tau'(b)} Z_i \geq b\right\} \\
    &= \sum_{t=1}^\infty P_0\left\{\sum_{i=1}^t Z_i \geq b, t = {\tau'(b)}\right\} \\
    &\leq \sum_{t=1}^\infty P_0\left\{\sum_{i=1}^t Z_i \geq b\right\} \\
    &= \sum_{t=1}^\infty P_0\left\{\sum_{i=1}^t (Z_i+\Delta) \geq b+t\Delta\right\} \\
    &\stackrel{(i)}{\leq} \sum_{t=1}^\infty \exp \left(-\frac{(b+t\Delta)^2}{2(t \sigma_0^2 + M (b+t\Delta))}\right) \\
    &\stackrel{(ii)}{\leq} \int_0^\infty \exp \left(-\frac{(b+x \Delta)^2}{2(x \sigma_0^2 + M (b+x \Delta))}\right) d x \nonumber\\
    &= a \int_0^\infty \exp \left(-\frac{(a \Delta y + C)^2}{2y}\right) d y \\
    &= a e^{-a \Delta C} \int_0^\infty e^{-((a^2 \Delta^2 / 2)y + (C^2/2) y^{-1})} d y \\
    &\stackrel{(iii)}{=} \frac{2 C}{\Delta} e^{-a \Delta C} K_1(a \Delta C)
\end{align*}
where $a := (\sigma_0^2+M\Delta)^{-1}$ and $C := \sigma_0^2 b / (\sigma_0^2 + M \Delta)$. In the series of inequalities above, $(i)$ follows from Bernstein's inequality \cite[p. 9]{all-nonpara-stat}, $(ii)$ follows from bounding the sum with an integral, and $(iii)$ follows from Lemma \ref{int_bessel} in the appendix, with $u=a^2 \Delta^2 / 2$ and $v=C^2/2$. Since $K_1(z) = \sqrt{\frac{\pi}{2 z}} e^{-z} (1+o(1))$ as $|z| \to \infty$,
the asymptotic result follows. \qedhere
\end{proof}


\begin{theorem}
\label{mct:fa}
Under the same assumptions as in Lemma \ref{mct:cross_fa}, let $\td{b}'_\alpha$ be such that
\begin{equation}
\label{thr_set}
    \sqrt{\frac{2 \pi \sigma_0^2 \td{b}'_\alpha}{\Delta^3}} \exp{\left(- \frac{2 R_0^2 \Delta}{\sigma_0^2} \td{b}'_\alpha \right)} = \alpha.
\end{equation}
Then, the MCT with $\td{b}'_\alpha$, i.e., $\tau(\td{b}'_\alpha)$, meets the FAR constraint \eqref{eq:FAR_MCT} asymptotically 
as $\alpha \to 0$.

Furthermore, as $\alpha \to 0$,
\begin{equation}
\label{thr_set2}
    \td{b}'_\alpha = \frac{\td{b}_\alpha}{R_0^2} (1+o(1))
\end{equation}
where $\td{b}_\alpha$ is defined in \eqref{balpha} and $R_0$ is defined in \eqref{rdef}.
\end{theorem}

\begin{proof}
As $\alpha \to 0$, $\td{b}'_\alpha \to \infty$. Recall the definition of $\tau'(b)$ in \eqref{tau_prime}. From  Lemma~\ref{mct:cross_fa}, for any $P_{1,t} \in {\cal M}_1$, $P_0\left\{S_{\tau'(\td{b}'_\alpha)} \geq \td{b}'_\alpha\right\} \leq \alpha (1+o(1))$. Then, using \cite[Sec.~2.6]{siegmund_1985}, it can be shown that
\begin{equation}
    \E{\infty}{P_0,\bm{P_1}}{\tau(\td{b}'_\alpha)} = \frac{\E{}{P_0}{\tau'(\td{b}'_\alpha)}}{P_0\left\{S_{\tau'(\td{b}'_\alpha)} \geq \td{b}'_\alpha \right\}} 
    \stackrel{(*)}{\geq} \frac{1}{P_0\left\{S_{\tau'(\td{b}'_\alpha)} \geq \td{b}'_\alpha \right\}} \geq \alpha^{-1}(1+o(1))
\end{equation}
where $(*)$ follows because $\E{}{P_0}{\tau'(\td{b}'_\alpha)} \geq 1$. Thus, \eqref{eq:FAR_MCT} is satisfied asymptotically.

For the second result, it is sufficient to show that $(\td{b}'_\alpha - \td{b}_\alpha)/\td{b}_\alpha = R_0^{-2}-1+o(1)$. Let
\begin{equation}
    D := \frac{2 \Delta}{\sigma_0^2} \td{b}'_\alpha - |\ln \alpha|.
\end{equation}
Then, recalling the definition of $\td{b}_\alpha$ in \eqref{balpha}, we have
\begin{equation}
\label{51}
    \frac{\td{b}'_\alpha - \td{b}_\alpha}{\td{b}_\alpha} = \frac{2 \Delta \td{b}'_\alpha}{|\ln{\alpha}| \sigma_0^2} - 1 = \frac{D}{|\ln{\alpha}|}
\end{equation}
and we need to show that
\begin{equation}
\label{D_rate}
    D = (R_0^{-2}-1)|\ln{\alpha}|+o(|\ln{\alpha}|).
\end{equation}
Rearranging the terms in \eqref{51}, we can express $\td{b}'_\alpha$ as:
\begin{equation}
    \td{b}'_\alpha = \td{b}_\alpha \left(1 + \frac{D}{|\ln{\alpha}|} \right) = \frac{\sigma_0^2}{2 \Delta} \left(|\ln{\alpha}| + D \right).
\end{equation}
Plugging this expression for $\td{b}'_\alpha$ into \eqref{nng:fa}, we have
\begin{equation}
    \sqrt{\frac{\sigma_0^4 \pi}{\Delta^4} (D+|\ln\alpha|)} e^{-R_0^2 (D+|\ln \alpha|)} = \alpha.
\end{equation}
Taking log on both sides, we obtain
\begin{equation}
\label{D_log_alpha}
    -\frac{1}{2}\ln \left(\frac{\sigma_0^4 \pi}{\Delta^4} (D+|\ln\alpha|)\right) + R_0^2 (D+|\ln \alpha|) = |\ln\alpha|.
\end{equation}

In the following, we first hypothesize that $D = D_1 |\ln \alpha| + o(|\ln \alpha|)$, where $D_1$ is not a function of $\alpha$, and then validate the hypothesis. Using this expression of $D$, the first term becomes
\begin{align*}
    \ln \left(\frac{\sigma_0^4 \pi}{\Delta^4} (D+|\ln\alpha|)\right) & = \ln \left(\frac{\sigma_0^4 \pi}{\Delta^4} ((D_1+1) |\ln \alpha| + o(|\ln \alpha|)\right) \\
    & = o(|\ln\alpha|).
\end{align*}
Therefore, \eqref{D_log_alpha} can be restated as:
\begin{equation}
    D = (R_0^{-2}-1) |\ln \alpha| + o(|\ln \alpha|).
\end{equation}
This validates our hypothesis on $D$, and also establishes \eqref{D_rate}. The proof is now complete.
\end{proof}

\begin{remark}
The threshold $\td{b}'_\alpha$ that meets the FAR constraint \eqref{eq:FAR_MCT} asymptotically can be obtained by solving \eqref{thr_set} numerically. Alternatively, we can use the approximation in \eqref{thr_set2} along with \eqref{balpha} to set:
\begin{equation} \label{eq:thrsh_set_3}
\td{b}'_\alpha = \frac{\td{b}_\alpha}{ R_0^2} = \frac{\sigma_0^2 |\ln \alpha|}{ 2 R_0^2 \Delta}.
\end{equation}
\end{remark}

\subsubsection{Worst-case Delay Analysis}
We now turn to the delay analysis of MCT. The following two lemmas are useful in establishing the delay performance. Specifically, Lemma~\ref{sprt-finite} is used to guarantee that MCT statistic is finite in expectation, Lemma~\ref{walds_non_stat} is used to extend Wald's identity to the non-stationary setting, and finally Theorem~\ref{mct:delay} is used to upper bound the asymptotic delay of MCT in the case where $P_{1,t}$'s are non-stationary.

\begin{lemma}
\label{sprt-finite}
Suppose that $P_{1,t} \in {\cal M}_1$ for all $t \geq 1$. Then, for any $b > 0$, $\E{1}{P_0,\bm{P_1}}{\tau(b)} < \infty$.
\end{lemma}

\begin{lemma}
\label{walds_non_stat}
Let $Z_1,Z_2,\dots$ be independent random variables. For any $t \geq 1$, $Z_t \sim P_t$ and $\E{}{P_t}{Z_t} \geq \Delta$. Let $T$ be any stopping time w.r.t. $Z_1,Z_2,\dots$ such that $\E{}{\bm{P}}{T} < \infty$. Then,
\begin{equation}
    \E{}{\bm{P}}{\sum_{t=1}^T Z_t} \geq \E{}{\bm{P}}{T} \Delta.
\end{equation}
\end{lemma}

The proofs of the lemmas are given in the appendix. Using these lemmas, we can upper bound the asymptotic delay as follows.
\begin{theorem}
\label{mct:delay}
Under the same assumptions as in Lemma \ref{mct:cross_fa}, the worst-case delay satisfies
\begin{equation}
\label{est_stat}
    \sup_{\bm{P_1}: P_{1,t} \in \mathcal{M}_1,\forall t} \WADD{P_0,\bm{P_1}}{\tau(\td{b}'_\alpha)} = \frac{|\ln{\alpha}| \sigma_0^2}{2\Delta^2 R_0^2}(1+o(1)) 
\end{equation}
as $\alpha \to 0$, where $\td{b}'_\alpha$ is defined in \eqref{thr_set}.
\end{theorem}
\begin{proof}
Following Lemma \ref{sprt-finite}, the MCT stopping time is finite in expectation even when the post-change distributions are non-stationary (but lie in ${\cal M}_1$). Thus, for any $P_{1,t} \in {\cal M}_1$,
\begin{align}
    \WADD{P_0,\bm{P_1}}{\tau(\td{b}'_\alpha)} &\leq \E{1}{P_0,\bm{P_1}}{\tau(\td{b}'_\alpha)} \nonumber\\
    &\stackrel{(i)}{\leq} \frac{1}{\Delta} \E{}{\bm{P_1}}{\sum_{t=1}^{\tau(\td{b}'_\alpha)} Z_t}\nonumber\\
    &= \frac{1}{\Delta} \E{}{\bm{P_1}}{\sum_{t=1}^{\tau(\td{b}'_\alpha)-1} Z_t + Z_{\tau(\td{b}'_\alpha)}} \nonumber\\
    &\stackrel{(ii)}{\leq} \frac{1}{\Delta} \left(\Tilde{b}'_\alpha + 1\right)\nonumber\\
    &= \frac{|\ln{\alpha}| \sigma_0^2}{2\Delta^2 R_0^2}(1+o(1))
\end{align}
where $(i)$ follows by Lemma \ref{walds_non_stat}, and $(ii)$ follows because $Z_{\tau(\td{b}'_\alpha)} \leq 1$. Thus,
\begin{equation}
    \sup_{\bm{P_1}: P_{1,t} \in \mathcal{M}_1, \forall t \geq 1} \WADD{P_0,\bm{P_1}}{\tau_{\td{b}'_\alpha}} \leq \frac{1}{\Delta} \left(\Tilde{b}'_\alpha + 1\right) = \frac{|\ln{\alpha}| \sigma_0^2}{2\Delta^2 R_0^2}(1+o(1))
\end{equation}
where $o(1) \to 0$ as $\alpha \to 0$.

For the other direction, consider stationary $P_{1,t} = P_1^* \in {\cal M}_1$ with the post-change mean $\E{}{P_1^*}{X_i} = \eta$, which implies $\E{}{P_1^*}{Z_i} = \Delta$. Then, as $\alpha \to 0$,
\begin{align}
    \WADD{P_0,P_1^*}{\tau(\td{b}'_\alpha)} &= \frac{\td{b}_\alpha'}{\Delta} (1+o(1)) \nonumber\\
    &= \frac{|\ln{\alpha}| \sigma_0^2}{2\Delta^2 R_0^2}(1+o(1))
\end{align}
where the first line follows by a standard renewal theory argument \cite[Sec. 2.5]{tartakovsky_sequential}. \qedhere
\end{proof}



\begin{remark}
As $\Delta \to 0$, $R_0 \to 1$. Thus, the result in Theorem~\ref{mct:delay} becomes
\begin{equation}
    \sup_{\bm{P_1}: P_{1,t} \in \mathcal{M}_1, \forall t \geq 1} \WADD{P_0,\bm{P_1}}{\tau(\td{\Lambda}^{\mu_0,\eta}, \td{b}'_\alpha)} = \frac{|\ln{\alpha}| \sigma_0^2}{2\Delta^2}(1+o(1))
\end{equation}
where $o(1)$ goes zero as $\alpha$ and $\Delta$ go to zero, which coincides with the minimax robust worst-case delay in (\ref{est_opt_delay}).
\end{remark}





\section{Numerical Results and Discussion}
\label{sec:num-res}
 We study the performance of the proposed tests through simulations for the case where the pre- and post-change distributions are Beta(4,16) ($\mu_0 = 0.2$) and Beta(4.5,16) ($\mu_1 = 0.2195$), respectively. The mean-threshold $\eta$ is set to be $0.21$. In particular, we compare the performances for the following three test statistics:
\begin{enumerate}
    \item The CuSum statistic for the case where both the pre- and post-change distributions are known, defined in (\ref{cusum_stat}).
    \item The statistic when only the pre-change distribution is known,  defined in (\ref{opt_stat}).
    \item The MCT statistic defined in (\ref{stat}).
\end{enumerate}
For all three statistics, based on their recursive structure, it is easy to show that the worst-case value of the change-point for computing WADD in \eqref{orig_qcd} is $\nu=1$. Therefore we can estimate the worst-case delays of the tests by simulating the post-change distribution from time 1.

\begin{figure}[tbp]
\centerline{\includegraphics[width=.75\textwidth,height=9cm]{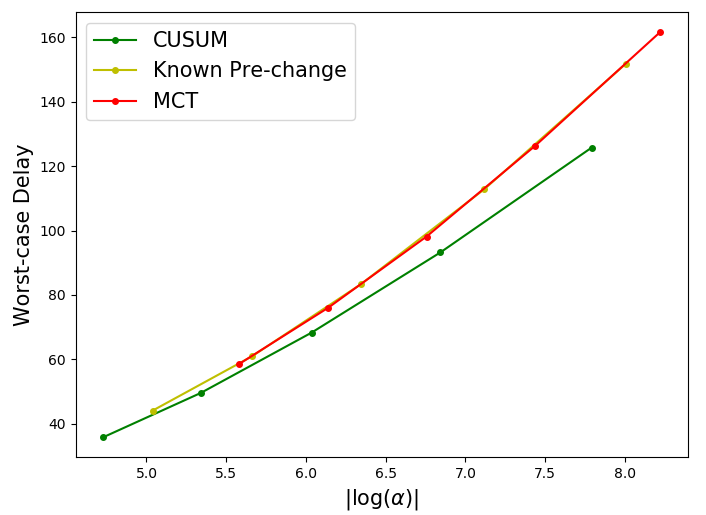}}
\vspace{-3mm}\caption{Performances of tests with different statistics. The pre- and post-change distribution are Beta(4,16) ($\mu_0 = 0.2$) and Beta(4.5,16) ($\mu_1 = 0.2195$), respectively. The mean-threshold $\eta = 0.21$.}
\label{fig:cusum_mct}
\end{figure}

We see in Fig. \ref{fig:cusum_mct} that the performance of MCT is very close to that of the asymptotically minimax robust optimal test that uses the full knowledge of the pre-change distribution. Note that the MCT statistic uses only the pre-change mean; the variance is required for setting the threshold to meet a given FAR constraint. 

\begin{figure}[tbp]
\centerline{\includegraphics[width=.75\textwidth,height=9cm]{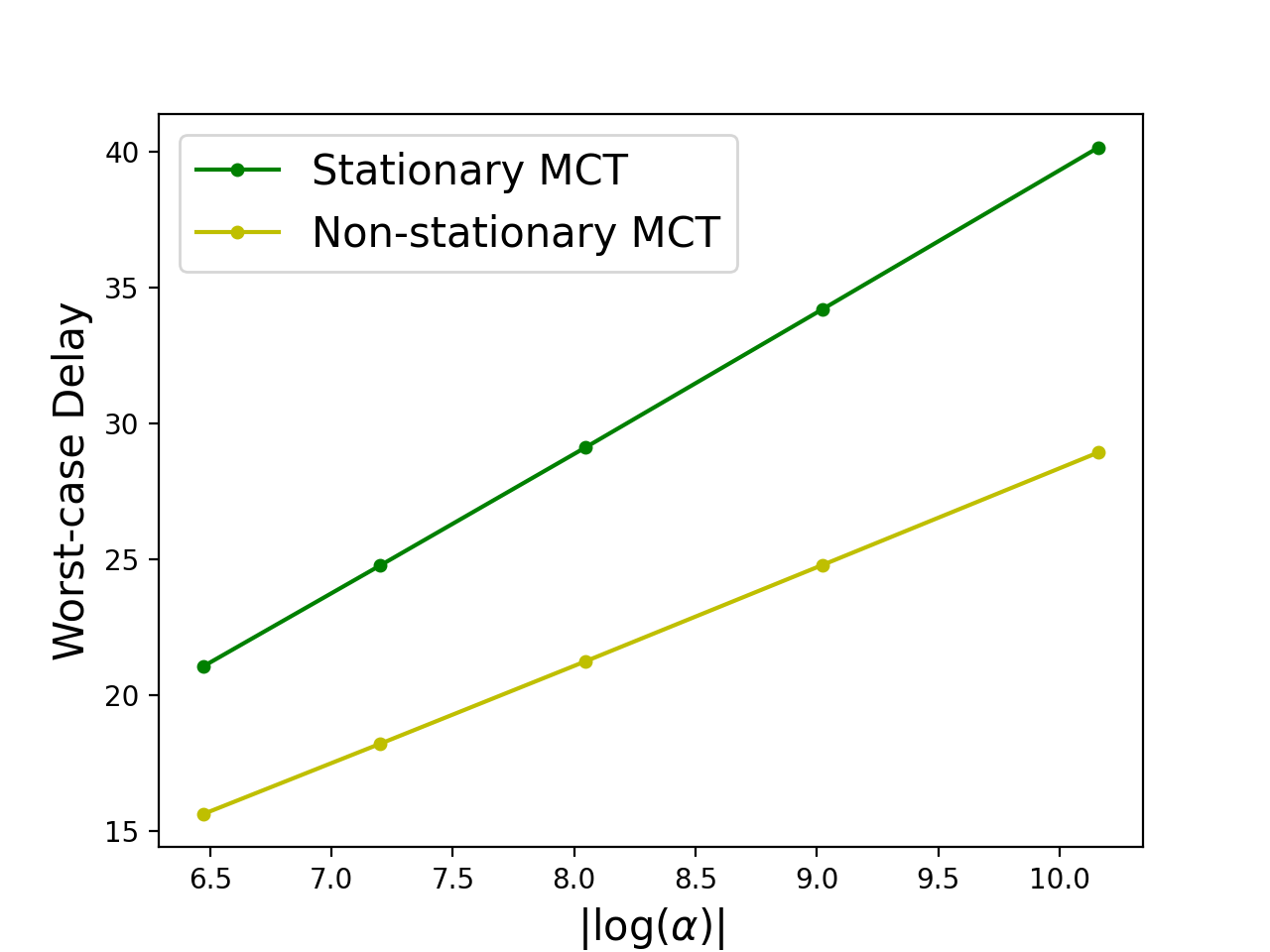}}
\vspace{-3mm}\caption{Performances of MCT under stationary and non-stationary post-change distributions. In the stationary case, the pre- and post-change distribution are Beta(2,2) ($\mu_0 = 0.5$) and Beta(3.5,2) ($\mu_1 = 0.636$), respectively, and the mean-threshold $\eta = \mu_1$. In the non-stationary case, the post-change observations are drawn from Beta($A$,2) at each time $t$, where $A \sim$ Unif(3.5,4.5).}
\label{fig:mct_non_stat}
\end{figure}

In Fig. \ref{fig:mct_non_stat}, we compare the performance of the MCT when the post-change distribution is non-stationary with that when the post-change distribution is stationary, for beta distributed observations. In the stationary case, we choose the post-change distribution to have mean $\mu_1 = \eta$, and in the non-stationary we choose the post-change distributions such that they all have mean greater than or equal to $\eta$. We observe, as expected, that the worst-case delay in the non-stationary case is always smaller than that in the stationary case. 

\begin{figure}[tbp]
\centerline{\includegraphics[width=.75\textwidth,height=9cm]{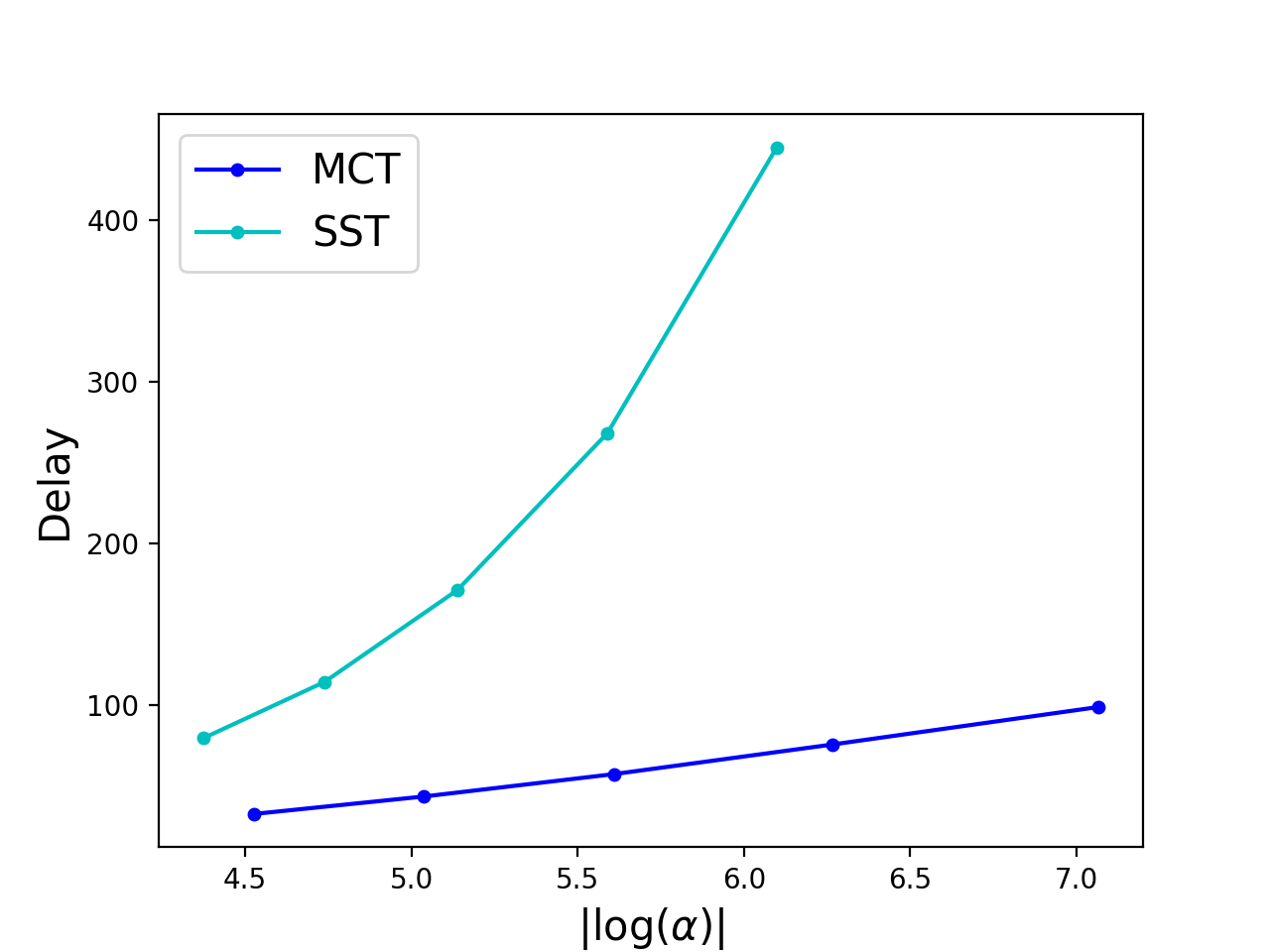}}
\vspace{-3mm}\caption{Performances of MCT and SST ($\tau_{\text{scan}}$ as defined in \eqref{scan-stat}). The pre- and post-change distribution are Beta(4,16) ($\mu_0 = 0.2$) and Beta(4.5,16) ($\mu_1 = 0.2195$), respectively.}
\label{fig:scan}
\end{figure}

Now, we compare our MCT test with a test using scan statistics (without windowing), defined as (see, e.g., \cite{maillard19a}):
\begin{equation}
\label{scan-stat}
    \tau_{\text{scan}}(b) := \inf \{t: \exists s \in [2,t]: \abs{\hat{\mu}_{1:s-1} - \hat{\mu}_{s:t}} \geq b \}
\end{equation}
where, assuming $s \leq t$,
\begin{equation}
    \hat{\mu}_{s:t} := \frac{1}{t-s+1} \sum_{i=s}^t X_i.
\end{equation}
The scan statistic test (SST) $\tau_{\text{scan}}$ is designed to detect a change in the mean of the observation sequence, but does not incorporate the knowledge that the post-change mean is greater than or equal to $\eta$. The SST also does not require knowledge of the pre-change mean, but it requires the change-point to be large enough so that a reasonable estimate of the pre-change mean can be obtained from $\hat{\mu}_{1:s-1}$.  

In the results shown in Fig.~\ref{fig:scan}, we assume that the change-point occurs after the first 100 observations are collected.
To allow for a fair comparison between MCT and SST, we use the first 100 observations to estimate $\mu_0$ for use in the MCT statistic, instead of assuming that $\mu_0$ is known. For the MCT simulation, the statistic is initialized after the estimation of $\mu_0$ from the first 100 samples, and therefore the delay is simulated by assuming that the change happens immediately after initialization, which corresponds to $\nu=1$, the worst-case value of the change-point. For the SST simulation, the change-point is set  $\nu=101$, which may not necessarily result in the worst-case delay. In Fig. \ref{fig:scan}, we see that the worst-case delay for MCT is much smaller than the delay of $\tau_{\text{scan}}$ at $\nu=101$, which is a lower bound of the worst-case delay of $\tau_{\text{scan}}$ over all possible change-points.

\begin{figure}[htbp]
\centerline{\includegraphics[width=\textwidth,height=10cm]{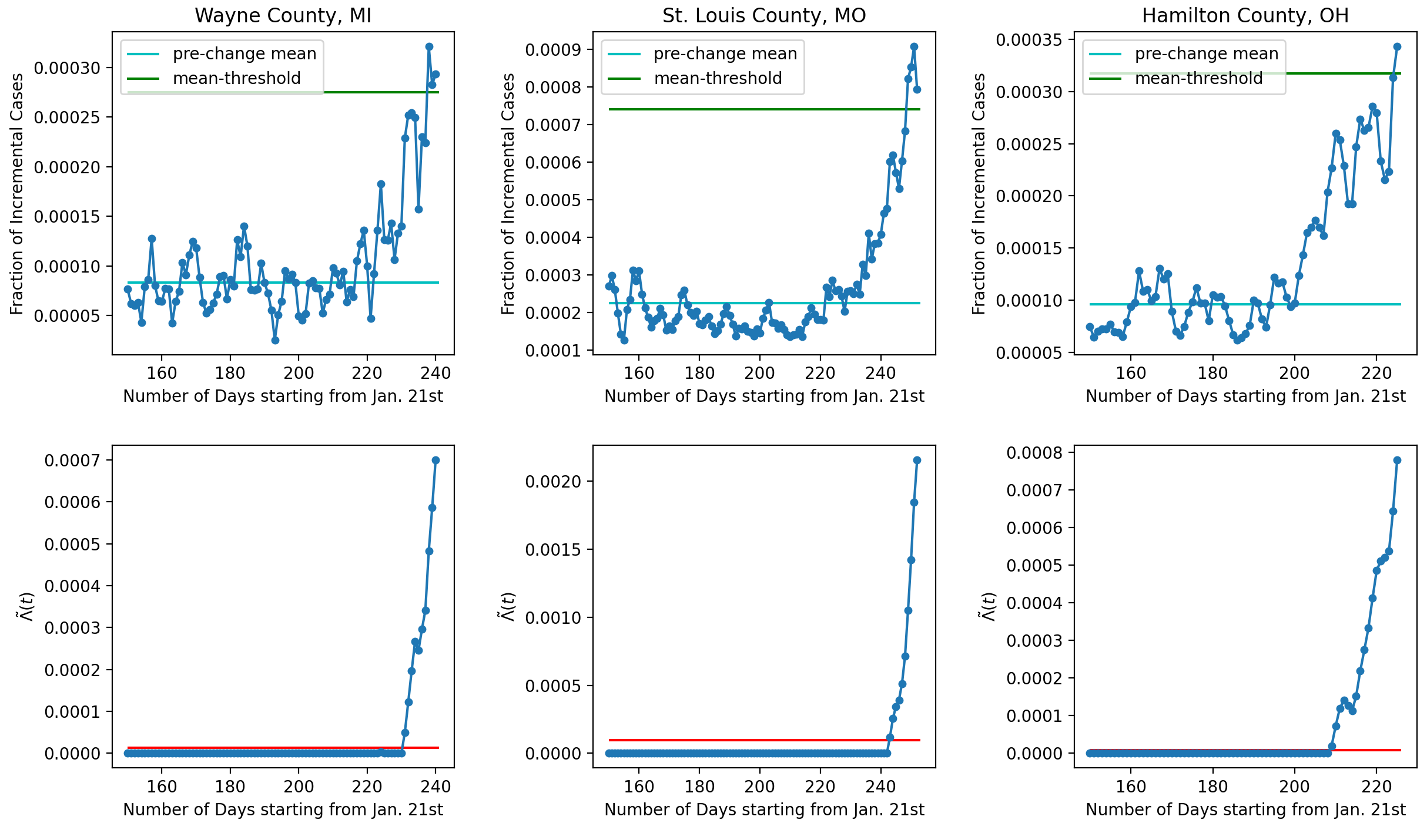}}
\vspace{-3mm}\caption{COVID-19 monitoring example. The upper subplot is the three-day moving average of the new cases  of COVID-19 as a fraction of the population in Wayne County, MI (left), St. Louis County, MO (middle), and Hamilton County, OH (right). The x-axis is the number days elapsed after January 21, 2020. The pre-change mean and variance are estimated using data from days 120 to 150. The FAR threshold $\alpha$ is set to $0.01$. For each county, the mean-threshold $\eta$ (in green) is set to be 3.3 times of the estimated pre-change mean (in cyan). The lower subplot shows the evolution of the statistic $\Tilde{\Lambda}$ in the corresponding county. The $\Lambda$-threshold $\td{b}_\alpha$ (in red) is calculated using equation (\ref{balpha}).}
\label{fig:covid}
\end{figure}

In Fig. \ref{fig:covid}, we apply the MCT to monitoring the spread of COVID-19 using new case data from various counties in the US \cite{nyt-covid-data}. The incremental cases from day to day can be assumed to be roughly independent.  The goal is to detect the onset of a new wave of the pandemic based on the incremental cases as a fraction of the county population exceeding some pre-specified level. 
The pre-change mean and variance are estimated using observations for periods in which the increments remain low and roughly constant. We set the mean-threshold $\eta$ to be a multiple of the pre-change mean, with understanding that such a threshold might be indicative of a new wave.  With this choice, we observe that 
the MCT statistic significantly and persistently crosses the test-threshold around late November in all counties, which is strong indication of a new wave of the pandemic. More importantly, unlike the raw observations which are highly varying, the MCT statistic shows a clear dichotomy between  the pre- and post-change settings, with the statistic staying near zero before the purported onset of the new wave, and taking off nearly vertically after the onset.


\section{Conclusion}
\label{sec:concl}

We studied the problem of quickest detection of a change in the mean of an observation sequence to a value above a pre-specified threshold in a non-parametric setting, allowing for the post-change distribution to be non-stationary.  For the case where the pre-change distribution is known, we derived a test that asymptotically minimizes the worst-case detection delay over all post-change distributions, as the false alarm rate goes to zero. In the process of deriving this asymptotically optimal test, we provided some new results for the general problem of asymptotic minimax robust quickest change detection in non-stationary settings, which should be of independent interest.  We then studied the limiting form of the optimal test as the gap between the pre- and post-change means goes to zero, the MCT. The MCT statistic only requires knowledge of the pre-change mean. Under the additional assumption that the distributions of the observations have bounded support, we derived an asymptotic upper bound on the FAR of the MCT for moderate values of mean gap, which can be used to set the threshold of the MCT using only knowledge of the pre-change mean and variance. We also characterized the asymptotic worst-case delay of the MCT for moderate values of the mean gap. 

We validated our analysis through numerical results for detecting a change in the mean of a beta distribution. In particular, we found that the MCT suffers little performance loss relative to the asymptotically optimal test with known pre-change distribution. We also showed that the MCT can significantly outperform tests based on prior work on scan statistics, which do not use information about the post-change mean threshold $\eta$. We also demonstrated the use of the MCT for detecting the onset of a new wave of an existing pandemic.

A possible avenue for future research on this topic is the detection of a change in statistics other than the mean. It is also of interest to study the mean change detection problem in sensor network settings.

\bibliographystyle{IEEEtran}
\bibliography{ref}

\appendix
The following lemma is useful for the proof of Lemma~\ref{molloy_delay_ext}.
\begin{lemma}
\label{wlln_non_stat}
Let $Y_1,\dots,Y_n$ be independent, zero-mean random variables. Suppose 
\begin{equation*}
    \sup_{1 \leq t \leq n} \E{}{}{Y_t^2} = o(n)    
\end{equation*}
as $n \to \infty$. Then, as $n \to \infty$,
\begin{equation*}
    n^{-1} \sum_{i=1}^n Y_i \xrightarrow{p.} 0.
\end{equation*}
\end{lemma}

\begin{proof}[\textbf{Proof of Lemma \ref{wlln_non_stat}}]
Denote $U_n = \sum_{i=1}^n Y_i$. By Chebyshev's inequality, for any $\eps > 0$,
\begin{align*}
    \Prob{}{}{\abs{n^{-1} \sum_{i=1}^n Y_i} > \eps} &= \Prob{}{}{U_n^2 > (n \eps)^2}\\
    &\leq \frac{\E{}{}{U_n^2}}{n^2 \eps^2}\\
    &\stackrel{(*)}{\leq} \frac{n \max_{1 \leq i \leq n} \E{}{}{Y_i^2}}{n^2 \eps^2} \xrightarrow{n \to \infty} 0
\end{align*}
where $(*)$ is due to the fact that $Y_i$'s are independent with zero-mean. 
\end{proof}

\begin{proof}[\textbf{Proof of Lemma \ref{molloy_delay_ext}}]
Fix $0 < \delta < 1$. Denote $\tau_b(\td{P}_0,\td{P}_1)$ as a short-hand notation for $\tau(\Lambda^{\td{P}_0,\td{P}_1}, b)$.
%
For any $t \geq \nu$, let
\begin{equation}
I_{t}^{\td{P}_0,\td{P}_1}
     := \E{}{P_{1,t}}{L^{\td{P}_0,\td{P}_1}(X)} = \KL{P_{1,t}}{\td{P}_0} - \KL{P_{1,t}}{\td{P}_1}.
\end{equation}
By definition of WS boundedness,
\begin{equation}
I_{t}^{\td{P}_0,\td{P}_1} \geq \KL{\td{P}_1}{\td{P}_0}.
\end{equation}
Let $Z_{t_1}^{t_2} (\bm{P_0},\bm{P_1}) := \sum_{t=t_1}^{t_2} \ln L^{P_{0,t},P_{1,t}} (X_t)$. Let $n_c := \floor{b/(1-\delta)\KL{\td{P}_1}{\td{P}_0}}$. 

From the proof of Theorem 4 in \cite{lai1998} (and also Theorem 1 in \cite{Molloy2017}), if we can establish
\begin{equation}
\label{wlln2}
    \lim_{n \to \infty} \Prob{\nu}{\bm{P_0},\bm{P_1}}{n^{-1} Z_t^{t+n-1} (\bm{\td{P}_0},\bm{\td{P}_1}) \leq \KL{\td{P}_1}{\td{P}_0} (1-\delta)} = 0
\end{equation}
for $t > \nu$, then, with a large enough $b$, we can get a large enough $n_c$ to satisfy
\begin{equation}
    \Prob{\nu}{\bm{P_0},\bm{P_1}}{n_c^{-1} Z_t^{t+n_c-1} (\bm{\td{P}_0},\bm{\td{P}_1}) \leq \KL{\td{P}_1}{\td{P}_0} (1-\delta)} < \delta,
\end{equation}
or equivalently,
\begin{equation}
    \Prob{\nu}{\bm{P_0},\bm{P_1}}{Z_t^{t+n_c-1} (\bm{\td{P}_0},\bm{\td{P}_1}) < b} < \delta.
\end{equation}
By independence (despite post-change being non-stationary), we get
\begin{align}
    &\quad \esssup \Prob{\nu}{\bm{P_0},\bm{P_1}}{(\tau_b(\td{P}_0,\td{P}_1)-\nu+1)^+ > t n_c | {\cal F}_{\nu-1} } \nonumber\\
    &\leq \esssup \Prob{\nu}{\bm{P_0},\bm{P_1}}{Z_{\nu+(j-1)n_c}^{\nu+j n_c - 1} (\bm{\td{P}_0},\bm{\td{P}_1}) < b,\forall 1 \leq j \leq t | {\cal F}_{\nu-1} } \leq \delta^t
\end{align}
for any $\nu \geq 1$ and $t \geq 1$. Therefore,
\begin{align}
    &\quad \esssup \E{\nu}{\bm{P_0},\bm{P_1}}{(\tau_b(\td{P}_0,\td{P}_1)-\nu+1)^+|{\cal F}_{\nu-1}} \nonumber\\
    &\leq n_c \sum_{t=1}^\infty \Prob{\nu}{\bm{P_0},\bm{P_1}}{n_c^{-1} (\tau_b(\td{P}_0,\td{P}_1)-\nu+1)^+ > t} \nonumber\\
    &\leq n_c \sum_{t=0}^\infty \delta^t = \frac{n_c}{1-\delta},
\end{align}
and from the definition of $n_c$,
\begin{equation}
    \WADD{\bm{P_0},\bm{P_1}}{\tau_b(\td{P}_0,\td{P}_1)} \leq (1+o(1))\left(\frac{b}{\KL{\td{P}_1}{\td{P}_0}}\right) \frac{1}{(1-\delta)^2}.
\end{equation}
Because $\delta$ is arbitrary, we can take $\delta \to 0$ and the proof is complete.

It remains to show (\ref{wlln2}). For any $t > \nu$ and $\delta > 0$,
\begin{align}
    &\quad \Prob{\nu}{\bm{P_0},\bm{P_1}}{k^{-1} \sum_{i=t}^{t+k-1} \ln L^{\td{P}_0,\td{P}_1} (X_i) \leq (1-\delta) \KL{\td{P}_1}{\td{P}_0}} \nonumber\\
    &\stackrel{(*)}{\leq} \Prob{\nu}{\bm{P_0},\bm{P_1}}{k^{-1} \sum_{i=t}^{t+k-1} \ln L^{\td{P}_0,\td{P}_1} (X_i) \leq k^{-1} \sum_{i=t}^{t+k-1} I_{i}^{\td{P}_0,\td{P}_1} -\delta \KL{\td{P}_1}{\td{P}_0} } \nonumber\\
    &= \Prob{\nu}{\bm{P_0},\bm{P_1}}{k^{-1} \sum_{i=t}^{t+k-1} \underbrace{\left(\ln L^{\td{P}_0,\td{P}_1} (X_i) - I_{i}^{\td{P}_0,\td{P}_1} \right)}_{\text{zero mean, independent}} \leq -\delta \KL{\td{P}_1}{\td{P}_0}}.
\end{align}
Note that $(*)$ follows from the WS boundedness assumption, and $\delta \KL{\td{P}_1}{\td{P}_0}$ is some strictly positive constant. Next, we will use the previous lemma. Denote $Y_i = \ln L^{\td{P}_0,\td{P}_1} (X_i) - I_{i}^{\td{P}_0,\td{P}_1}$.
Thus, by Lemma \ref{wlln_non_stat},
\begin{equation}
    k^{-1} \sum_{i=t}^{t+k-1} \ln L^{\td{P}_0,\td{P}_1} (X_i) - I_{i}^{\td{P}_0,\td{P}_1} \xrightarrow{p.} 0,
\end{equation}
and thus (\ref{wlln2}) is proved. Now the proof is complete. \qedhere
\end{proof}


\begin{proof}[\textbf{Proof of Lemma \ref{molloy_fa_ext}}]
Recall that if no change ever happens, $X_t \sim P_{0,t} \in {\cal P}_0$ for all $t \geq 1$ and $\bm{P_0} = \{P_{0,t}\}_{t \geq 1}$. Here $P_{0,t}$ could be non-stationary. We follow the procedure in \cite[Thm.~4]{lai1998}. For simplicity, denote $Y_t \equiv \ln L^{\td{P}_0,\td{P}_1}(X_t)$. 

Define the stopping times:
\begin{equation}
    \sigma_{m+1} := \inf \left\{t > \sigma_m: \sum_{i=\sigma_m+1}^t Y_i < 0 \right\},
\end{equation}
and let $\sigma_0 := 0$ and $\inf \emptyset := \infty$. Suppose for now that we can establish that, on $\{\sigma_m < \infty\}$,
\begin{equation}
\label{superm2}
    \bm{P_0}\left\{\sum_{i=\sigma_m+1}^t Y_i \geq b \text{ for some } t > \sigma_m \bigg| {\cal F}_{\sigma_m} \right\} \leq e^{-b}
\end{equation}
for any threshold $b > 0$. Define the number of zero-crossings before hitting the threshold as
\begin{equation}
    M := \inf \left\{ m \geq 0: \sigma_m < \infty \text{ and } \sum_{i=\sigma_m+1}^t Y_i \geq b \text{ for some } t > \sigma_m \right\}.
\end{equation}
Thus, for any $m > 0$,
\begin{align}
    \bm{P_0}\{M > m\} &= \E{}{\bm{P_0}}{\ind{M > m}} \nonumber\\ 
    &= \E{}{\bm{P_0}}{\ind{M > m \text{ and } M > m-1}}\nonumber\\
    &= \E{}{\bm{P_0}}{\bm{P_0}\{M > m | {\cal F}_{\sigma_m} \}\ind{M > m-1}} \nonumber\\
    &= \E{}{\bm{P_0}}{\bm{P_0}\left\{\sum_{i=\sigma_m+1}^t Y_i < b \text{ for any } t > \sigma_m \bigg| {\cal F}_{\sigma_m} \right\}\ind{M > m-1}} \nonumber\\
    &\geq (1-e^{-b})\bm{P_0}\{M > m-1 \} \nonumber\\
    &\geq (1-e^{-b})^m
\end{align}
where the first inequality follows from \eqref{superm2} and the second one follows from recursion. Therefore,
\begin{equation}
    \E{}{\bm{P_0}}{\tau(\Lambda^{\td{P}_0,\td{P}_1},b)} \geq \E{}{\bm{P_0}}{M} \geq \sum_{m=0}^\infty (1-e^{-b})^m = e^b.
\end{equation}

It remains to show \eqref{superm2}. By WS boundedness condition, $\E{}{P_{0,t}}{\exp(Y_t)} \leq 1$ for any $t \geq 1$. Therefore, $\{\exp\left(\sum_{i=k}^n Y_i\right),{\cal F}_n,n \geq k\}$ is a non-negative supermartingale under $\bm{P_0}$, and
\begin{align}
    \bm{P_0}\left\{\sum_{i=\sigma_m+1}^{t} Y_i \geq b \text{ for some } t > \sigma_m \bigg| {\cal F}_{\sigma_m} \right\} &\leq \bm{P_0}\left\{\max_{\sigma_m+1 \leq n \leq t} \sum_{i=\sigma_m+1}^n Y_i \geq b \bigg| {\cal F}_{\sigma_m} \right\} \nonumber\\
    &\stackrel{(*)}{\leq} e^{-b}\  \E{}{P_{0,\sigma_m+1}}{\exp(Y_{\sigma_m+1})} \nonumber\\
    &\leq e^{-b}
\end{align}
where $(*)$ follows from Lemma 1 in \cite{Molloy2017}. The proof is now complete. \qedhere
\end{proof}

\begin{proof}[\textbf{Proof of Theorem \ref{asymp_opt_non_stat}}]
The proof steps are similar to \cite{Molloy2017}. From max-min inequality, it is true that
\begin{multline}
    \inf_{T \in {\cal C}_\alpha^{{\cal P}_0}}\quad \sup_{(\bm{P_0},\bm{P_1}): (P_{0,t},P_{1,t}) \in {\cal P}_0 \times {\cal P}_1 ,\forall t} \WADD{\bm{P_0},\bm{P_1}}{T} \\\geq \sup_{(\bm{P_0},\bm{P_1}): (P_{0,t},P_{1,t}) \in {\cal P}_0 \times {\cal P}_1 ,\forall t} \quad \inf_{T \in {\cal C}_\alpha^{{\cal P}_0}} \WADD{\bm{P_0},\bm{P_1}}{T}.
\end{multline}
It suffices to prove the other direction.

For any $(\bm{P_0},\bm{P_1})$ such that $(P_{0,t},P_{1,t}) \in {\cal P}_0 \times {\cal P}_1$ for any $t \geq 1$, we have
\begin{align}
\label{minmax_proof}
    \WADD{\bm{P_0},\bm{P_1}}{\tau(\Lambda^{\td{P}_0,\td{P}_1},b_\alpha)} &\stackrel{(i)}{\leq} (1+o(1))\left(\frac{b_\alpha}{\KL{\td{P}_1}{\td{P}_0}}\right) \nonumber\\
    &\stackrel{(ii)}{=} \WADD{\td{P}_0,\td{P}_1}{\tau(\Lambda^{\td{P}_0,\td{P}_1},b_\alpha)}\nonumber\\
    &\stackrel{(iii)}{=} \inf_{T \in {\cal C}^{\td{P}_0}_\alpha} \WADD{\td{P}_0,\td{P}_1}{T} \nonumber\\
    &\stackrel{(iv)}{=} \inf_{T \in {\cal C}^{{\cal P}_0}_\alpha} \WADD{\td{P}_0,\td{P}_1}{T} \nonumber\\
    &\stackrel{(v)}{\leq} \sup_{(\bm{P_0},\bm{P_1}): (P_{0,t},P_{1,t}) \in {\cal P}_0 \times {\cal P}_1 ,\forall t} \quad \inf_{T \in {\cal C}^{{\cal P}_0}_\alpha} \WADD{\bm{P_0},\bm{P_1}}{T}
\end{align}
where $o(1) \to 0$ as $\alpha \to 0$. In the above series of inequalities, $(i)$ follows directly from Lemma \ref{molloy_delay_ext}, $(ii)$ and $(iii)$ follow from standard CuSum analyses (e.g., \cite{lai1998}), $(iv)$ is justified below, and $(v)$ follows from the fact that $(\td{P}_0,\td{P}_1) \in {\cal P}_0 \times {\cal P}_1$. Note that $(iii)-(v)$ are satisfied for any $0 < \alpha < 1$.

We now justify $(iv)$. Since $\td{P}_0 \in {\cal P}_0$, ${\cal C}^{{\cal P}_0}_\alpha \subseteq {\cal C}^{\td{P}_0}_\alpha$. Following standard CuSum analysis (e.g., \cite{lai1998}), $\tau(\Lambda^{\td{P}_0,\td{P}_1},b_\alpha) \in {\cal C}^{\td{P}_0}_\alpha$. From Lemma \ref{molloy_fa_ext}, for any $P_{0,t} \in {\cal P}_0$, $\FAR{\bm{P_0}}{\tau(\Lambda^{\td{P}_0,\td{P}_1},b_\alpha)} \leq \alpha$, and therefore $\tau(\Lambda^{\td{P}_0,\td{P}_1},b_\alpha) \in {\cal C}^{{\cal P}_0}_\alpha$. For any $\alpha$, since $\tau(\Lambda^{\td{P}_0,\td{P}_1},b_\alpha)$ achieves the infimum over the set ${\cal C}^{\td{P}_0}_\alpha$, it also does over the subset ${\cal C}^{{\cal P}_0}_\alpha$.

Since \eqref{minmax_proof} holds for any $(\bm{P_0},\bm{P_1}):(P_{0,t},P_{1,t}) \in {\cal P}_0 \times {\cal P}_1,\forall t \geq 1$,
\begin{multline}
    \sup_{(\bm{P_0},\bm{P_1}): (P_{0,t},P_{1,t}) \in {\cal P}_0 \times {\cal P}_1 ,\forall t} \WADD{\bm{P_0},\bm{P_1}}{\tau(\Lambda^{\td{P}_0,\td{P}_1},b_\alpha)} \\\leq \sup_{(\bm{P_0},\bm{P_1}): (P_{0,t},P_{1,t}) \in {\cal P}_0 \times {\cal P}_1 ,\forall t} \quad \inf_{T \in {\cal C}^{{\cal P}_0}_\alpha} \WADD{\bm{P_0},\bm{P_1}}{T},
\end{multline}
and thus
\begin{multline}
    \inf_{T \in {\cal C}_\alpha^{{\cal P}_0}}\quad \sup_{(\bm{P_0},\bm{P_1}): (P_{0,t},P_{1,t}) \in {\cal P}_0 \times {\cal P}_1 ,\forall t} \WADD{\bm{P_0},\bm{P_1}}{T} \\\leq \sup_{(\bm{P_0},\bm{P_1}): (P_{0,t},P_{1,t}) \in {\cal P}_0 \times {\cal P}_1 ,\forall t} \quad \inf_{T \in {\cal C}^{{\cal P}_0}_\alpha} \WADD{\bm{P_0},\bm{P_1}}{T}.
\end{multline}
Therefore, $\tau(\Lambda^{\td{P}_0,\td{P}_1},b_\alpha)$ asymptotically solves \eqref{newmainprob} as $\alpha \to 0$, and
\begin{equation}
    \sup_{(\bm{P_0},\bm{P_1}): (P_{0,t},P_{1,t}) \in {\cal P}_0 \times {\cal P}_1 ,\forall t} \WADD{\bm{P_0},\bm{P_1}}{\tau(\Lambda^{\td{P}_0,\td{P}_1},b_\alpha)} = \left(\frac{|\ln \alpha|}{\KL{\td{P}_1}{\td{P}_0}}\right)(1+o(1))
\end{equation}
where $o(1) \to 0$ as $\alpha \to 0$.\qedhere
\end{proof}

The following Lemma is useful for the proof of Lemma \ref{mct:cross_fa}.
\begin{lemma}
\label{int_bessel}
Let $u,v$ be some constant. Then,
\begin{equation}
    \int_0^\infty \exp\left(-\left(u y + \frac{v}{y}\right)\right) d y = 2 \sqrt{\frac{v}{u}} K_1 (2 \sqrt{u v})
\end{equation}
where $K_\beta (z)$ is the modified Bessel function of the second kind of order $\beta$.
\end{lemma}

\begin{proof}
Let $y = e^\theta \sqrt{v/u}$. Then, the integral becomes
\begin{align}
    \int_0^\infty \exp\left(-\left(u y + \frac{v}{y}\right)\right) d y &= \int_{-\infty}^\infty \exp\left(-\sqrt{u v}\left(e^\theta + e^{-\theta}\right)\right) \sqrt{v/u} e^{\theta} d \theta  \nonumber\\
    &= \sqrt{v/u} \int_{-\infty}^\infty \exp\left(-2\sqrt{u v}\cosh(\theta) \right) e^{\theta} d \theta  \nonumber\\
    &= \sqrt{v/u} \int_{-\infty}^\infty \exp\left(-2\sqrt{u v} \cosh(\theta) \right) (\cosh(\theta)+\sinh(\theta)) d \theta \nonumber\\
    &\stackrel{(*)}{=} 2 \sqrt{v/u} \int_{0}^\infty \exp\left(-2\sqrt{u v} \cosh(\theta) \right) \cosh(\theta) d \theta \nonumber\\
    &= 2 \sqrt{v/u} K_1 (2 \sqrt{u v})
\end{align}
where $(*)$ follows because $\cosh(\theta)$ is an even function while $\sinh(\theta)$ is an odd function.
\end{proof}

\begin{proof}[\textbf{Proof of Lemma \ref{sprt-finite}}]
Recall that $Z_i := X_i - (\mu_0+\eta)/2$ and $S_t = \sum_{i=1}^t Z_i$. By assumption on ${\cal M}_1$, let $Z_t$ have mean $\Delta_t \geq \Delta$ for any $t$ under measure $P_{1,t}$. Fix $b > 0$. Define the supplementary stopping time
\begin{equation}
    \Bar{\tau}'(b) := \inf\left\{t \geq 1: S_t \geq b \right\}.
\end{equation}
Consider $t > t_0 := \floor{b / \Delta}$. Then,
\begin{align}
    \Prob{1}{P_0,\bm{P_1}}{\Bar{\tau}'(b) > t} &= \bm{P_1}\{\Bar{\tau}'(b) > t\} \nonumber\\
    &= \bm{P_1}\left\{\sum_{i=1}^t Z_i < b\right\} \nonumber\\
    &= \bm{P_1}\left\{\sum_{i=1}^t (Z_i - \Delta) < b - t \Delta \right\} \nonumber\\
    &\leq \bm{P_1}\left\{\abs{\sum_{i=1}^t (Z_i - \Delta_i)} > t \Delta - b\right\} \nonumber\\
    &\stackrel{(*)}{\leq} 2 \exp \left(-\frac{2 \Delta^2 (t - b / \Delta)^2}{t}\right)
\end{align}
where $(*)$ follows from Hoeffding's inequality.

Using the same technique as the proof of lemma \ref{mct:cross_fa},
\begin{equation}
    \int_0^\infty \exp \left(-\frac{2 \Delta^2 (t - b / \Delta)^2}{t}\right) d t = \frac{2 b}{\Delta} e^{4 b \Delta} K_1(4 b \Delta) < \infty
\end{equation}
where $K_1(z)$ is the modified Bessel function of the second kind with order 1. Hence,
\begin{align}
    \E{}{\bm{P_1}}{\Bar{\tau}'(b)} &= \sum_{t=0}^\infty \bm{P_1}\left\{\Bar{\tau}'(b) > t\right\} \nonumber\\
    &\leq \sum_{t=0}^{t_0} \bm{P_1}\left\{\Bar{\tau}'(b) > t\right\} + \int_{t_0}^\infty 2 \exp \left(-\frac{2 \Delta^2 (t - b / \Delta)^2}{t}\right) d t \nonumber\\
    &< \infty.
\end{align}
Therefore, for any $P_{1,t} \in {\cal M}_1$, $\E{1}{P_0,\bm{P_1}}{\Bar{\tau}'(b)} = \E{}{\bm{P_1}}{\Bar{\tau}'(b)} < \infty$. Finally, it follows directly that $\E{1}{P_0,\bm{P_1}}{\tau(b)} \leq \E{1}{P_0,\bm{P_1}}{\Bar{\tau}'(b)} < \infty$. \qedhere
\end{proof}

\begin{proof}[\textbf{Proof of Lemma \ref{walds_non_stat}}]
For each $t \geq 1$, let $Z_t^+ := \max\{0,Z_t\}$ and $Z_t^- := -\min\{0,Z_t\}$. Note that $Z_t^+,Z_t^- \geq 0$ and $Z_t = Z_t^+ - Z_t^-$. Therefore,
\begin{equation}
    \E{}{\bm{P}}{\lim_n \sum_{t=1}^n Z_t^+ \ind{t \leq T}} = \lim_n \sum_{t=1}^n \E{}{\bm{P}}{Z_t^+ \ind{t \leq T}}
\end{equation}
by Monotone Convergence Theorem, since $\sum_{t=1}^n Z_t^+ \ind{t \leq T}$ is non-decreasing in $n$. The same argument applies to $Z_t^-$. Hence,
\begin{align}
    \E{}{\bm{P}}{\sum_{t=1}^T Z_t} &= \E{}{\bm{P}}{\sum_{t=1}^T Z_t^+} - \E{}{\bm{P}}{\sum_{t=1}^T Z_t^-} \nonumber\\
    &=\E{}{\bm{P}}{\sum_{t=1}^\infty Z_t^+ \ind{t \leq T}} - \E{}{\bm{P}}{\sum_{t=1}^\infty Z_t^- \ind{t \leq T}} \nonumber\\
    &= \sum_{t=1}^\infty \E{}{\bm{P}}{Z_t^+ - Z_t^-} \E{}{\bm{P}}{\ind{t \leq T}} \nonumber\\
    &\geq \Delta \sum_{t=1}^\infty \E{}{\bm{P}}{\ind{T \geq t}} \nonumber\\
    &= \E{}{\bm{P}}{T} \Delta.
\end{align}
\end{proof}

\end{document}